\documentclass{article}
\usepackage{waingarten}
\usepackage{graphicx}
\usepackage{tikz}
\usepackage{stmaryrd}
\usepackage{tcolorbox}

\algdef{SE}[SUBALG]{Indent}{EndIndent}{}{\algorithmicend\ }%
\algtext*{Indent}
\algtext*{EndIndent}

\usetikzlibrary{arrows.meta, positioning, shapes.geometric, fit, quotes, calc}

\newcommand{\Dec}{\mathrm{Dec}}

\newcommand{\hypercube}{\{0,1\}^d}
\newcommand{\linspan}{\mathrm{span}}

\newcommand{\Ham}{\mathrm{Ham}}

\providecommand{\email}[1]{\texttt{#1}}

\title{A Near-Optimal Space Lower Bound for Euclidean Diameter Estimation in Dynamic Streams}

\author{
  Ashwin Padaki\thanks{Supported by the National Science Foundation (NSF) GRFP under Grant No. DGE-2236662, and Grant No. CCF-2337993. \{\email{apadaki@seas.upenn.edu}\}}  
  \and 
  Krish Singal\thanks{Supported by the National Science Foundation (NSF) under Grant No. CCF-2337993 \{\email{ksingal@seas.upenn.edu}\}} \\[3mm]University of Pennsylvania
  \and 
  Erik Waingarten\thanks{Supported by the National Science Foundation (NSF) under Grant No. CCF-2337993 \{\email{ewaingar@seas.upenn.edu}\}} 
}

\begin{document} 

\maketitle 

\begin{abstract}
We study the space complexity of diameter estimation for a set of points in Euclidean space in the dynamic (turnstile) streaming model. The seminal work of Indyk~\cite{I03} gives a $c$-approximation to the Euclidean diameter of $n$ vectors using \smash{$n^{O(1/c^2)}$} space. Our main contribution is giving an essentially matching lower bound. Any dynamic streaming algorithm which can $c$-approximate the diameter of $n$ Euclidean vectors must use \smash{$n^{\tilde{\Omega}(1/c^2)}$} space.
\end{abstract} 
\newpage
\section{Introduction}

We study streaming algorithms for estimating the diameter of a set of points in a high-dimensional Euclidean space. Specifically, we work within Indyk's \emph{geometric streaming model}~\cite{I04b}, where an input stream $\sigma$ consists of insertions and deletions of points $x\in \{0, \dots, \Delta\}^d$. The stream $\sigma$ implicitly defines a (multi-)subset of $\{0,\dots, \Delta\}^d$, encoded by a frequency vector \smash{$v\in\Z_{\geq 0}^{\{0,\dots, \Delta\}^d}$}. We will study algorithms which can approximately compute the Euclidean diameter of the surviving point set:
 \[
    \diam(v) \;=\; \max\bigl\{\,\|x - y\|_2 \;:\; v_{x}, v_{y} > 0 \,\bigr\}.
\]
In the dynamic (or turnstile) streaming model, the vector $v$ is updated through a stream of insertions and deletions. Throughout, we let $N := |\{ 0, \ldots, \Delta\}^d| = (\Delta+1)^d$, and we let $n := \|v\|_0$ denote the number of surviving points at the end of the stream; all the upper and lower bounds in this paper will be in terms of $n$. This work will study the bit-complexity of a streaming algorithm as a function of $n$ and the desired approximation $c >  1$ for the gap-decision version of the problem: the streaming algorithm is instantiated for a fixed scale $r>0$ and approximation factor $c > 1$, and after observing a stream of insertions and deletions to the frequency vector $x$, the algorithm must distinguish whether $\diam(x) \le r$ or $\diam(x) \ge cr$.\footnote{The gap-decision version of a problem can always be solved given an algorithm to the estimation version, but the reverse direction naively incurs a multiplicative $\log(d\Delta)$-factor in the space complexity.}

In insertion-only streams, the space complexity of Euclidean diameter is relatively well-understood. Agarwal and Sharathkumar~\cite{AS15} first gave a $(\sqrt{2}+\eps)$-approximation in $O(d\log(1/\eps)/\eps^3)$ space, which was recently improved to $O(d\log(1/\eps) / \eps^2)$~\cite{HMM25}.\footnote{Note, we may assume $d = O(\log n / \eps^2)$ via the Johnson-Lindenstrauss lemma.} Furthermore,~\cite{AS15} gave a lower bound showing that any approximation ratio below $\sqrt{2}$ requires space growing linearly in $n$ or exponentially in $d$. The picture is a lot less clear in dynamic streams, and settling the space complexity was raised as an open question by Krauthgamer~\cite{K22}.

The best upper bound is implied by a work of Indyk~\cite{I03}. While that work was primarily focused on dynamic data structures for Euclidean approximate furthest neighbor, it leads to a dynamic streaming algorithm for Euclidean diameter using $\tilde{O}(n^{2/(c^2-1)})$ space.\footnote{In~\cite{I03}, the streaming result is technically stated for insertion-only streams with space $\tilde{O}(n^{1/(c^2-1)})$. To implement the algorithm in dynamic streams, one can use $\ell_0$-sampling (which was not developed until later~\cite{JST11}), and a careful analysis reveals that the complexity degrades slightly to \smash{$\tilde{O}(n^{2/(c^2-1)})$}. This quadratic increase comes from the need to union bound over the non-contraction/expansion of $O(n^2)$ pairs of points rather than $O(n)$ pairs (as in the insertion-only setting).} For ease of comparison, we include an exposition of Indyk's algorithm for diameter estimation in dynamic streams in Section~\ref{sec: upper-bound}. Since Indyk's algorithm, there have been no further improvements (either upper or lower bounds) on dynamic streaming algorithms for Euclidean diameter.

\subsection{Our Result} \label{sec:results}

Our main result is showing that Indyk's algorithm for diameter estimation in dynamic streams is essentially optimal---Indyk's algorithm uses $n^{O(1/c^2)}$ space for a $c$-approximation, and our lower bound shows $n^{\tilde{\Omega}(1/c^2)}$ space is necessary. In Theorem~\ref{thm:lb-intro} below, a randomized streaming algorithm $\calA$ for the $(r,c,n)$-diameter problem processes an arbitrary stream of insertions and deletions (of arbitrary length and arbitrary frequency) while maintaining an internal state and access to public randomness. The space complexity of $\calA$ is given by $\calS^+(\calA, m)$ and is a function of $m \in \N$, encoding the logarithm of the distinct number of final internal states that $\calA$ may hold after processing an arbitrary stream whose final frequency vector $v \in \Z^{\{0,\dots, \Delta\}^d}$ has $v_x \in \{0,\dots, m\}$ for all $x \in \{0,\dots, \Delta\}^d$ (see Section~\ref{sec: prelims}).

\begin{restatable}{theorem}{StreamLb}\label{thm:ln-sketch-lb}\label{thm:lb-intro}
    Any randomized streaming algorithm $\calA$ which solves the $(r,c,n)$-diameter problem using bounded randomness complexity has $\calS^+(\calA, m) = n^{\tilde{\Omega}(1/c^2)} \cdot \log(m+1)$ for all $m \in \N$.
\end{restatable}

As we elaborate on in the technical overview, the primary challenge in establishing Theorem~\ref{thm:lb-intro} is proving a lower bound on the dimensionality of linear sketches solving the $(r,c,n)$-diameter problem. We use a recent framework of~\cite{KPSW25}, who study diameter estimation in general (non-Euclidean) metrics and gave sketching-to-streaming reductions for scale-invariant problems building on prior work of~\cite{LNW14}. We emphasize that in Theorem~\ref{thm:lb-intro}, the hypothesis is algorithm $\calA$ solves the $(r,c,n)$-diameter problem over any streams (in particular, streams of arbitrary length and of arbitrary magnitude) and utilizes randomness complexity which can be an arbitrary function of $r, n, c$, $d$ and $\Delta$ but importantly not the stream length---all of these assumptions are satisfied by Indyk's algorithm. Up to these conditions, we consider Theorem~\ref{thm:lb-intro} a strong negative answer to Krauthgamer's question~\cite{K22}.

\subsection{Related Work}
\label{sec:related}

\textbf{Space-approximation tradeoffs in geometric streaming.} Many geometric streaming problems in high dimensions exhibit a tradeoff where the required space is polynomial in $n$, with an exponent that decays with the approximation factor $c$. Examples include Euclidean minimum spanning tree ($n^{O(1/\sqrt{c})}$ space~\cite{CJLW21, CCJLW23}), Euclidean facility location ($n^{O(1/c)}$ space \cite{CJKVY22}), and Earth Mover's Distance over a two-dimensional $[\Delta]\times [\Delta]$ grid (using $\Delta^{O(1/c)}$ space~\cite{ABIW09}). However, only substantially weaker lower bounds are known:~\cite{AIK08} shows that $\Omega(d/c)$ (or equivalently $\Omega(\log n / c)$) bits are necessary for a $c$-approximation to the Earth Mover's distance over $(\{0,1\}^d, \ell_1)$. The work of~\cite{CJLW21} adapted a qualitatively similar lower bound for streaming the minimum spanning tree over $(\{0,1\}^d, \ell_1)$, and~\cite{AKR15} shows that the Earth Mover's distance over $[\Delta]\times [\Delta]$ cannot be sketched with constant bits and constant approximation. This work is the first to give substantially stronger lower bounds for any (natural) high-dimensional Euclidean streaming problem.

\textbf{Reductions from streaming to linear sketching.} An important technique used is a reduction from streaming to linear sketching from~\cite{KPSW25}, building on a line-of-work first developed by~\cite{G08, LNW14} and further studied by~\cite{KP20}. The reduction from~\cite{KPSW25} requires the streaming algorithm correctly compute on \emph{arbitrarily long} streams whose frequency vectors have \emph{arbitrary magnitude} and have bounded-randomness complexity. After the work of \cite{LNW14}, there has been interest in reducing the length of the necessary streams~\cite{HLY19, JLY26}. In particular, the very recent~\cite{JLY26} gave substantial progress by giving streaming-to-sketching reductions with polynomial-length streams (for bounded-magnitude frequency vectors). At the moment, it is unclear to us whether we can apply~\cite{JLY26} to rule out streaming algorithms for the diameter on polynomially-long streams.\footnote{The reason is because here and in~\cite{KPSW25}, the streaming algorithm $\calA$ is assumed to work for unbounded streams with unbounded magnitudes, and the space complexity is measured as a function of the magnitude of frequency vectors; $\calS^+(\calA, m)$ is the number of bits needed to specify every final internal state reachable by arbitrary streams whose final frequency vector has magnitude at most $m$. Importantly, the output linear sketch $(T, \Dec)$ from~\cite{KPSW25} correctly computes on any distribution over frequency vectors (even those with very large magnitude), while the dimensionality and granularity of $T$ depend only on $\calS^+(\calA, 1)$. This means that the magnitude $m$ of the frequency vectors in our distribution depends on the dimensionality and granularity of $T$. In~\cite{JLY26}, the parameter $m$ is fixed in advance, and the dimensionality of the linear sketch depends on the space complexity used at the fixed parameter $m$ (preventing us from increasing $m$ while keeping the matrix dimension and granularity fixed).}

\subsection{Technical Overview} \label{sec: tech-overview}

Our hard distribution $\calD$ in Theorem~\ref{thm:ln-sketch-lb} has frequency vectors $\bv$ whose support $\bS = \supp(\bv)$ will be fully contained in the hypercube $\hypercube$. Its structure is similar to that of \cite{KPSW25} in that $\bS$ is supported either on a tight ball of points, or on a tight ball and one far point. Any linear sketch which can solve $(r,c,n)$-diameter (Definition~\ref{def:r-c-n-diam}) must successfully distinguish whether $\bS$ has (Hamming) radius at most $d/c^2$ or at least $\Omega(d)$. In particular, consider generating $\bS$ by the following process:
\begin{enumerate}
\item Sample $\bx,\by \sim \{0,1\}^d$ uniformly and independently. Every point $z \in \calB(\bx,d/c^2)$ (i.e., whose Hamming distance from $\bx$ is at most $d/c^2$) is included in $\bS$ with some (arbitrary for now) multiplicity $\bv_z >0$.
\item With probability $1/2$, include $\by$ in $\bS$ with some (arbitrary for now) multiplicity $\bv_{\by}>0$, and with probability $1/2$, do not include $\by$ and set $\bv_{\by} = 0$.
\end{enumerate}
In the above instances, the diameter of $\bS$ is fully determined by whether or not the vector $\by$ is included in $\bS$. Moreover, there is multiplicative gap of $\Omega(c)$ in the Euclidean diameter between the two cases; note, the pointset $\bS$ has support size $n = \exp(\Theta(d\log c/c^2))$. Before discussing lower bounds, our starting point will be a natural algorithmic (upper bound) approach for these instances via linear sketching:
\begin{quote}
\textbf{Ball-Subspace-Evading Configurations.} Suppose we find an embedding $f \colon \{0,1\}^d \to \R^s$, such that with high probability over $\bx,\by \sim \{0,1\}^d$, the vector $f(\by)$ is not contained within $\Span_f(\calB(\bx,d/c^2))$, i.e., the span of vectors $f(z)$ for $z \in \calB(\bx, d/c^2)$.
\end{quote}
Then, $f$ leads to a linear sketch which maintains the $s$-dimensional vector $T\bv = \sum_{z \in \bS} \bv_z \cdot f(z)$ by letting $T$ have columns given by $f(z)$. Since $f(\by)$ is usually not in $\Span_f(\calB(\bx,d/c^2))$, the presence or absence of $f(\by)$ affects the subspace containing $\bv$. Perhaps surprisingly, the techniques of~\cite{KPSW25} allow us to show that any algorithmic approach gives rise to a linear sketch of the above form (see Lemma~\ref{lem:stream-to-sketch} and Lemma~\ref{lem:alg-to-embedding}). The primary technical contribution of this work is showing that any ball-subspace-evading configuration $f$ must map to a large dimension $s$.

\textbf{Comparison to~\cite{KPSW25}.} The recent work of~\cite{KPSW25} on streaming lower bounds for (non-Euclidean) diameter, consists of three steps. First, a streaming-to-sketching reduction for scale-invariant functions (Theorem~4 of~\cite{KPSW25}) which allows them to lower bound the dimensionality of linear sketches. Second, a distribution over inputs whose correctness on a low-dimensional linear sketch implies a low-rank factorization of a (random) zero-pattern matrix (Lemma~4.6 of~\cite{KPSW25}). Lastly, they apply a result of~\cite{ABGMM20} which shows that any factorization with the (random) zero-pattern must have high rank, a ``min-rank'' lower bound. In~\cite{KPSW25}, the random metric is constructed to exactly yield the random zero-pattern, which is crucial, since the proof of~\cite{ABGMM20} is via the probabilistic method; however, it leads to a non-Euclidean metric construction. In this work, we can use the first and second step of~\cite{KPSW25}, but cannot use the min-rank lower bound of~\cite{ABGMM20}. The novelty in this work is proving the dimensionality lower bound for ball-subspace-evading configurations, our notion analogous to min-rank. The new argument avoids the probabilistic method, which is important, since we prove lower bounds against a specific zero-pattern (derived from the hypercube instance).

\paragraph{Dimensionality of Ball-Subspace-Evading Configurations.} The remainder of the technical overview is devoted to answering that question: what is the smallest target dimension $s \in \N$ for which there exists a ball-subspace-evading configuration $f \colon \{0,1\}^d \to \R^s$. An initial (incorrect) instinct is to conjecture that unstructured embeddings, using dimensionality $s$ set to $|\calB(0, d/c^2)| + 1$, are best possible; a random embedding $f$ into these many dimensions ensures $\Span_f(\calB(\bx, d/c^2))$ is always $|\calB(0,d/c^2)|$-dimensional, and the final $f(\by)$ has one ``available direction'' in $\R^s$ left. However, Indyk's algorithm (which is a linear sketch) does give a very structured and much lower-dimensional embedding, with \smash{$s = n^{O(1/c^2)} = |\calB(0, d/c^2)|^{O(1/c^2)}$}. In Lemma~\ref{lem:dim-lb-embedding}, we show that any embedding where \smash{$f(\by) \notin \Span_f(\calB(\bx, \tau))$} with probability at least $1-\delta$ must have dimensionality \smash{$\exp(\Omega(\tau^2 / d))$}, and setting $\tau = d / c^2$ and substituting \smash{$n = \exp(O(d \log c / c^2))$} gives the desired lower bound on $s$.

The first step in the proof of Lemma~\ref{lem:dim-lb-embedding} is translating the above formulation into a rank lower bound on a large $2^n \times 2^n$ matrix with a specified zero-pattern. The assumption of $f$ is that $\Span_f(\calB(\bx, \tau))$ often does not contain $f(\by)$, so if we let $g \colon \{0,1\}^d \to \R^s$ be a random map where $g(x)$ is drawn orthogonal to $\Span_f(\calB(x, \tau))$ and consider the rank-$s$ matrix \smash{$A$} where $\smash{A}_{y,x} = \langle f(y), g(x)\rangle$, we have:
\begin{itemize}
    \item Every $x \in \{0,1\}^d$ and $z \in \calB(x,\tau)$ has $A_{x,z} = 0$ (by the way we generated $g(z)$); so the ``diagonal band'' of $A$ corresponding to pairs in $\{0,1\}^d$ within Hamming distance $\tau$ is always zero. 
    \item Whenever $f(y) \notin \Span_f(\calB(x, \tau))$, which occurs on at least $(1-\delta) \cdot 2^{2d}$ pairs, $g(x)$ will have some correlation with $f(y)$ almost surely, so $(1-\delta)$-fraction of entries in $A$ are non-zero.
\end{itemize}
It will be useful later in the argument to ensure (i) all entries are non-negative, and (ii) a constant fraction of columns contains at least $(1-2\delta) \cdot 2^d$ non-zeros. Condition (i) is enforced by tensoring which makes the rank $s^2$ (and loses only a factor $2$ in the exponent), and condition (ii) by applying Markov's inequality. 

\paragraph{Forster's Transform and Duality to Isolate Rank.} Lower bounds on rank of a non-negative $2^n \times 2^n$ matrix with conditions on sign- and zero-patterns arise in communication complexity, and in this work, we use Forster's approach for sign-rank lower bounds~\cite{F02}. Specifically, we will first apply a Forster transform in order to obtain a rank-$s^2$ non-negative matrix $B$ with entries in $[0,1]$ which has the same zero-pattern as $A$, but ensures $\| B \|_F^2 \geq \Omega((1-2\delta) 2^{2d} / s^2)$ (Lemma~\ref{lem:fro-lb}). This helps translate an assumption on the zero-pattern into one about substantial Frobenius norm. Note, such a matrix $B$ is useful for isolating the rank; if $M$ is an arbitrary $2^n\times 2^n$ matrix where $M_{x,y} = 1$ whenever $y \notin \calB(x, \tau)$, 
\[ \|B\|_F^2 \leq \sum_{x, y \in \{0,1\}^d} B_{x,y} = \langle B, M\rangle \leq \| B \|_{\Tr} \cdot \|M \|_{2} \leq s \| B\|_F \cdot \|M\|_2, \]
where the first inequality uses the fact $B \in [0, 1]$, the second that $M_{x,y} = 1$ for any non-zero $B_{x,y}$, and then, the duality of trace and operator norms, as well as the fact the trace norm is at most square-root rank times the Frobenius norm. Substituting the Frobenius norm lower bound, 
\[ s^2 \gsim \max\left\{ \dfrac{2^d}{\|M\|_2} : M_{x,y} = 1 \text{ whenever $y \notin \calB(x,\tau)$} \right\}. \]
Hence, it suffices to construct a matrix $M$ of smallest spectral norm which is $1$ on every entry outside the ``diagonal band'' of pairs within Hamming distance $\tau$. We remark that the proof of Lemma~\ref{lem:fro-lb} is not entirely a black-box application of Forster's transform. The existence of a Forster transform (i.e., an invertible linear map such that after rescaling, vectors are in isotropic position) which we will apply to the $g$-vectors requires a non-degeneracy condition on the set of vectors, where a single $k$-dimensional subspace cannot contain more than a $k/s$-fraction of the points. When the condition fails, we find a $k$-dimensional subspace containing more than $k/s$-fraction of columns and recurse. The process eventually terminates, and while the number of columns contributing to $\|B\|_F^2$ may decrease, the magnitude of their contributions increases proportionally.

\paragraph{Constructing the Matrix $M$.} We will construct the matrix $M$ by directly ensuring its eigenvectors are the characters $\{ \chi_S \in \{-1,1\}^{2^d} : S \subset [d] \}$ where $(\chi_S)_x = (-1)^{\langle \ind_S, x\rangle}$ which satisfy $M_{x,y} = 1$ outside the band. Specifically, we seek coefficients $c_0, \dots, c_d \in \R$ where
\[ M = 1_{2^d \times 2^d} + \sum_{\ell=0}^{d} c_{\ell} \sum_{|S|=\ell} \chi_{S} \chi_{S}^{\intercal}, \]
which directly sets the eigenvalues to $2^d(1-c_0)$ and $-2^d c_1, \dots,-2^d c_d$. Furthermore, the fact we sum over all sets $|S| = \ell$ ensures entries of $M$ only depend on the Hamming distance of the underlying inputs, i.e., $M_{x,y}$ only depends on $|x\oplus y|$, and the specific value can be read off from a linear function of the coefficients $c \in \R^{d+1}$. Specifically, we let $K$ be the $(d+1)\times (d+1)$ matrix where the $(j,\ell)$-entry is $\alpha_{\ell}(j) = \sum_{|S|=\ell} (\chi_S)_z$ with $|z| = j$,\footnote{Note, this is invariant to the specific choice of $z \in \{0,1\}^d$ and only depends on $|z|=j$.} then
\[ M_{x,y} = e_x^{\intercal} M e_y = 1 + (Kc)_{|x\oplus y|}. \]
The constraint that $M_{x,y} = 1$ when $y \notin \calB(x, \tau)$ is equivalent to requiring $(Kc)_{j} = 0$ for all $j \in \{\tau+1, \dots, d\}$. In other words, we may construct $M$ by finding a setting of $c \in \R^{d+1}$ where $(Kc)_{j} = 0$ for all $j \in \{\tau+1, \dots, d\}$ while maintaining an upper bound on
\begin{align} 
\| M \|_{2} = 2^d \cdot \max\left\{ |1-c_0|, |c_1|, \dots, |c_d| \right\}. \label{eq:spectral}
\end{align}
In Lemma~\ref{lem:ext-polynomial}, we relate the above linear program to a problem on extremal polynomials. In particular, there are three observations:
\begin{itemize}
    \item $K^2 = 2^d \cdot I$, which will allow us to write $p = Kc$ and instead of varying $c \in \R^{d+1}$, we will vary $p \in \R^{d+1}$ where $p_{\tau+1} = \dots = p_d = 0$.
    \item The $\ell$-th row of $K$ is the evaluation of a specific polynomial of degree-$\ell$ at inputs $j \in \{ 0, \dots, d\}$ (corresponding to the columns) which are linearly independent. 
    \item As we vary $p \in \R^{d+1}$ with $p_{\tau+1} = \dots = p_d = 0$, the vectors $Kp \in \R^{d+1}$ vary over all univariate polynomial evaluations of $(q(0), \dots, q(d))$ where $q$ is of degree at most $\tau$.
\end{itemize}
The question in (\ref{eq:spectral}) then becomes one of minimizing the value of $\max\{ |2^d - q(0)|, |q(1)|, \dots, |q(d)| \}$ over all univariate polynomials $q$ of degree at most $\tau$. We finish off the proof in Lemma~\ref{lem:cr}, by constructing a polynomial which ensures the spectral norm of $M$ is at most $\exp(\Omega(\tau^2 / d))$. The construction actually comes from a translated and rescaled polynomial underlying the $O(\sqrt{n \log(1/\eps)})$ approximate degree of the $\mathsf{OR}$ function~\cite{BCdWZ99,BT22}. Namely, consider the smallest possible $\eps > 0$ where $\tau \gsim \sqrt{d \log(1/\eps)}$, which gives us $\eps = \exp\left(-\Omega(\tau^2/d)\right)$; the univariate polynomial $f$ used for approximating $\mathsf{OR}$ must evaluate to $1\pm \eps$ on $i \in \{1, \dots, d\}$ and $-1\pm\eps$ on $0$. Letting $p(i) = -(f(i) - 1)/\eps$ gives the desired polynomial. 
\section{Preliminaries} \label{sec: prelims}

\newcommand{\freq}{\mathrm{freq}}

In the dynamic geometric streaming model of~\cite{I04}, a stream $\sigma$ consists of an arbitrary sequence of insertions and deletions of points, and implicitly maintains a multi-set of points present. The stream may $\textsc{Add}(p)$ (adds the point $p$ to the current set) or $\textsc{Remove}(p)$ (which removes $p$ from the current set) to update the initially empty multi-set. The inserted and deleted points $p$ lie in the discrete geometric space $\{0,\dots, \Delta\}^d$, so an equivalent representation of the underlying multi-set is a vector \smash{$v \in \Z^{\{0,\dots, \Delta\}^d}_{\geq 0}$} (where each coordinate corresponds to a point in $\{0,\dots, \Delta\}^d$) maintaining the frequency counts of each point.\footnote{Our upper bounds work in the less restrictive (general) turnstile model where frequency counts may be negative, and our lower bounds work in the (more restrictive) strict turnstile model.} We use the notation $v = \freq(\sigma)$ to denote the frequency vector of a stream $\sigma$. With this representation, the diameter of a frequency vector $v$ is
\[ \diam(v) := \max\left\{ \| x - y \|_2 : x, y \in \supp(v) \right\}, \]
where $\supp(v)$ consists of non-zero coordinates $x \in \{0,\dots,\Delta\}^d$, i.e., where $v_x \neq 0$.
A randomized streaming algorithm $\calA$ processes each update in a stream one-by-one while keeping a (compressed) internal state and access to randomness. At the end of the stream, the algorithm produces an output which is a function of the internal state and the randomness solely. 
\begin{definition}[$(r,c,n)$-diameter problem]\label{def:r-c-n-diam}
For $r > 0$, $c \geq 1$ and $n \in \N$, a randomized streaming algorithm $\calA$ solves the $(r, c, n)$-diameter problem if after processing any stream $\sigma$ with $v = \freq(\sigma)$,
\begin{itemize}
\item If $\diam(v) \leq r$ and $|\supp(v)|\leq n$, $\calA(\sigma)$ outputs $0$ with probability at least $1-\delta$,
\item If $\diam(v) \geq c r$ and $|\supp(v)| \leq n$, $\calA(\sigma)$ outputs $1$ with probability at least $1-\delta$,
\end{itemize} 
Here, $\delta > 0$ is a sufficiently small constant. We refer to $\diam_{r,c,n}(v) \in \{0,1,*\}$ as the partial function specifying the two cases.
\end{definition}
The primary focus of this paper is proving a dimensionality and space-complexity lower bounds on linear sketches and streaming algorithms which can solve the $(r, c, n)$-diameter problem. We rely on a streaming-to-sketching reduction which was first developed in~\cite{G08, LNW14}, and then specialized to scale-invariant problems (problems whose scaling of the frequency vectors do not affect the output), like $\diam(v)$. The reduction is a crucial ingredient of this work as well, which receives a randomized streaming algorithm $\calA$, as well as an input distribution $\calD$ over frequency vectors, and outputs a linear map which can be used to solve the $(r,c,n)$-diameter problem and whose dimensionality will depend on the space complexity of $\calA$; this allows us to reduce to a clean linear algebraic problem.

\begin{definition}\label{def:ln-sketch}
For $s, \lambda \in \N$, an $(s, \lambda)$-linear sketch is a matrix-algorithm pair $(T, \Dec)$, where $T$ is an $s \times N$ integer matrix with $\norm{T}_{\infty} \leq \lambda$ and $\Dec\colon \Z^s \to \{0,1\}$ is a function. On input frequency vector \smash{$v \in \Z_{\geq 0}^{\{ 0, \ldots, \Delta\}^d}$}, the sketch outputs $\Dec(Tv)$ as its output.
\end{definition}

\begin{definition}[Space Complexity]\label{def:space}
For $m \in \N$, we let $\calW^{+}(\calA, m)$ be the set of final internal states $\calA$ may hold over all streams $\sigma$ with $\| \freq(\sigma)\|_{\infty} \leq m$ and all settings of randomness.\footnote{We assume randomness is public, and the algorithm does not need to use internal states to store its randomness.} The space complexity of $\calA$ is the function $\calS^+(\calA, \cdot) \colon \N \to \N$ given by
\[ \calS^{+}(\calA, m) \eqdef \lceil \log_2 | \calW^+(\calA, m) | \rceil. \] 
\end{definition}
 

\section{Lower Bound} \label{sec: lower-bound}

In this section, we prove Theorem~\ref{thm:ln-sketch-lb}, which gives a lower bound on the dimensionality of linear sketches. Then, we will use the streaming-to-sketching reduction from~\cite{KPSW25} to obtain the corresponding lower bound on the bit-complexity of dynamic streaming algorithms.

\begin{restatable}{theorem}{LnSketchLb}\label{thm:ln-sketch-lb}
    For $n, \lambda \in \N$, any $c \geq 1$, and sufficiently small constant $\delta > 0$, there exists $r > 0$ and a distribution $\calD$ on frequency vectors $\bv$ such that $\diam_{r,c,n}(\bv)\in\{0,1\}$ with high probability, and every $(s,\lambda)$-linear sketch $(T, \Dec)$ which satisfies
    \begin{align*}
        \Prx_{\bv \sim \calD}\left[ \Dec(T\bv) = \diam_{r,c,n}(\bv)\right] \geq 1 - \delta
    \end{align*}
    must have $s = n^{\tilde{\Omega}(1 / c^2)}$.
\end{restatable}

\paragraph{Description of the Hard Distribution.} The distribution is supported on the Boolean hypercube, so it will be convenient to use the Hamming distance. For $x,z\in\hypercube$, let $\Ham(x,z) := |\{i\in[d] : x_i\ne z_i\}|$ and $\calB(x,\tau) := \{z\in\hypercube : \Ham(x,z)\le \tau\}$.

\begin{definition}[Hard Distribution $\calD$]\label{def:hard}
For $n \in \N$, and $c \geq 1$, let $d = d(n,c) \in \N$ denote the largest positive integer where
\[ 1 + \sum_{j=0}^{d/(12 c^2)} \binom{d}{j} \leq n, \]
and let $\tau = d / (12 c^2)$ and $r = 2\sqrt{\tau} = \sqrt{d}/(c\sqrt{3})$. For fixed positive integers $\sfP, \sfU \geq 1$ to be set later, $\calD$ is the following distribution over frequency vectors $\bv$ whose coordinates correspond to points in $\{0,\dots,\Delta\}^d$:
\begin{itemize}
\item Draw $\bx, \by \sim \{0,1\}^{d}$, viewing $\{0,1\}^{d}$ as a subset of $\{0, \dots, \Delta\}^d$. 
\item Set $\bv_{\by} = \sfP$ with probability $1/2$ and $\bv_{\by} = 0$ otherwise. 
\item For each $z\in \calB(\bx, \tau) \subset \{0,1\}^{d}$, draw $\bv_z$ uniformly at random from $\{1, \dots, \sfP \sfU\}$, and set all remaining coordinates of $\bv$ to $0$.
\end{itemize}
\end{definition}

Notice that Definition~\ref{def:hard} is designed such that $\supp(\bv)$ always has size at most $n$ (by the setting of $d$). Furthermore, $\diam(\bv) \le r$ whenever $\bv_{\by}=0$ (since $\supp(\bv)\subseteq\calB(\bx,\tau)$), while $\diam(\bv)\ge \sqrt{d/3} \ge cr$ whenever $\bv_{\by} > 0$ (since $\Ham(\bx,\by)\ge d/3$, with probability $1-o(1)$ over uniform $\bx,\by$). Indeed, this immediately implies the first property of Theorem~\ref{thm:ln-sketch-lb}, that $\diam_{r,c,n}(\bv)\in\{0,1\}$ with probability $1-o(1)$.

\paragraph{Ball-Subspace-Evading Functions.} We now show that correctness of the linear sketch $(T,\Dec)$ on the specific distribution $\calD$ leads to a structural constraint on $T$. Namely, by specifying the values of $\sfP$ and $\sfU$ in Definition~\ref{def:hard}, we force the columns of the matrix $T$ to yield an $s$-dimensional embedding that satisfies a property we call ball-subspace-evasion. The second part of the argument is a proof that any ball-subspace-evading embedding must have high dimension.

\begin{definition}\label{def:bse}
A function $f \colon \{0,1\}^{d} \to \R^s$ is $(\delta,\tau)$-ball-subspace-evading if 
\begin{align*}
\Prx_{\bx,\by\sim\{0,1\}^{d}}\left[ f(\by) \notin \Span_f(\calB(\bx,\tau))\right] \geq 1 - \delta,
\end{align*}
where $\Span_f(\calB(\bx,\tau)) \subseteq \R^s$ denotes the subspace spanned by the vectors $\{ f(z) : z \in \calB(\bx,\tau)\}$.
\end{definition}

\begin{restatable}{lemma}{AlgToEmbedding}\label{lem:alg-to-embedding}
For any $\lambda\in\N$, there is a setting of $\sfP,\sfU\ge 1$ such that if $(T,\Dec)$ is an $(s,\lambda)$-linear sketch such that \begin{align}
    \Prx_{\bv\sim \calD}\left[ \Dec(T\bv) \neq \diam_{r,c,n}(\bv) \right] \leq \delta. \label{eq:comp}
\end{align}
    then one can construct an $(O(\delta),\tau)$-ball-subspace-evading function $f : \{0,1\}^d\to\R^s$.
\end{restatable}

We note that Lemma~\ref{lem:alg-to-embedding} has an analogous step appearing in Subsection~4.1 of~\cite{KPSW25}. In that work, the linear sketch $(T, \Dec)$ yields a solution to a minrank problem (which is then shown to require large dimensionality). In this work, we use $T$ to construct ball-subspace-evading functions (which, while similar in spirit to a specific minrank problem, is not exactly the same). Nevertheless, the proof of Lemma~\ref{lem:alg-to-embedding} is entirely analogous to~\cite{KPSW25} and is thus deferred to Appendix~\ref{sec: appendix}.

\paragraph{Ball-Subspace-Evasion Requires Large Dimension.} The main technical novelty of this work is showing that ball-subspace-evading functions $f$ require large dimensionality. This proof is the technical bulk of the paper and is presented in Section~\ref{sec:bse-lb}. After that is established, we can conclude Theorem~\ref{thm:ln-sketch-lb} fairly straightforwardly.

\begin{lemma}[Main Lemma, deferred to Section~\ref{sec:bse-lb}]\label{lem:dim-lb-embedding}
If $f \colon \{0,1\}^d \to \R^s$ is $(\delta, \tau)$-ball-subspace-evading, then $s = \exp\left(\Omega(\tau^2/d)\right)$. 
\end{lemma}

\begin{proof}[Proof of Theorem~\ref{thm:ln-sketch-lb}]
We may now combine Lemma~\ref{lem:alg-to-embedding} and Lemma~\ref{lem:dim-lb-embedding} to conclude Theorem~\ref{thm:ln-sketch-lb}. First, we fix the sketch magnitude parameter $\lambda\in\N$, which gives a setting of parameters $\sfP,\sfU$ from Lemma~\ref{lem:alg-to-embedding} and in turn specifies the hard distribution $\calD$ from Definition~\ref{def:hard}. By Lemma~\ref{lem:alg-to-embedding} on a $(s,\lambda)$-linear sketch $(T,\Dec)$ with success probability $1-\delta$ on $\calD$ gives an $(O(\delta),\tau)$-ball-subspace-evading function $f \colon \{0,1\}^{d} \to \R^s$. Finally, Lemma~\ref{lem:dim-lb-embedding} implies 
\[ s = 2^{\Omega(\tau^2 / d)} = 2^{\Omega(d / c^4)} = n^{\Omega(1/ (c^2 \log c))}, \]
since $n = 2^{O(d \log c /c^2)}$.
\end{proof}


\subsection{Proof of Lemma~\ref{lem:dim-lb-embedding}}\label{sec:bse-lb}

In this section, we prove our main lemma which lower bounds the dimensionality of $(\delta,\tau)$-ball-subspace-evading configurations. First, given $f \colon \{0,1\}^d \to \R^s$, we let $g \colon \{0,1\}^d \to \R^s$ denote the random function where for every $x \in \{0,1\}^d$, if $f(\by) \in \Span_f(\calB(x,\tau))$ with probability at least $2\delta$ over $\by \sim \{0,1\}^d$, we set $g(x) = 0 \in \R^s$; otherwise, let $g(x)$ be a random unit vector drawn from the space orthogonal to $\Span_f(\calB(x, \tau))$.\footnote{Even though the function $g$ and (subsequently) the matrix $A$ will be a random, our proof will proceed by immediately considering any fixed $g$ and $A$, so we keep unbolded.} Consider the $2^d \times 2^d$ matrix $A$ whose rows are parametrized by points $x, y \in \{0,1\}^d$, and we let
\[ A_{y,x} = \langle f(y)^{\otimes 2}, g(x)^{\otimes 2} \rangle,  \]
and a consequence of $(\delta, \tau)$-subspace-evading configurations gives the following holds almost surely over the draw of $g$:
\begin{itemize}
\item For every $x \in \{0,1\}^d$ and every $z \in \calB(x, \tau)$, the entry $A_{x,z} = 0$. This holds because every $g(z)$ is orthogonal to $\Span_f(\calB(z,\tau)) \ni f(x)$.
\item $A$ has rank $s^2$, because the tensored vectors $f(y)^{\otimes 2}, g(x)^{\otimes 2} \in \R^{s^2}$.
\item At most half of the columns are all zero. For the non-zero columns $x \in \{0,1\}^d$, there are at least $(1-2\delta)$-fraction of $\by \sim \{0,1\}^d$ where $A_{\by,x} > 0$.
\end{itemize}
For the third item, consider the event over the draw of $\bx \sim \{0,1\}^d$ that the probability that $f(\by) \notin \Span_f(\calB(\bx,\tau))$ is at least $1-2\delta$. First, this event occurs with probability at least $1/2$ over $\bx \sim \{0,1\}^d$ by Markov's inequality, so at most half of the columns have $g(x) = 0$. Consider a non-zero column $x \in \{0,1\}^d$ where the event does occur, and note that $\Span_f(\calB(x,\tau)) \neq \R^s$ because $f(y) \notin \Span_f(\calB(x,\tau))$. In these columns, $g(x)$ is a random unit vector orthogonal to $\Span_f(\calB(x,\tau))$, and hence will have non-zero correlation with any vector $f(y) \notin \Span_f(\calB(x,\tau))$ almost surely. In other words, at least $(1-2\delta)$-fraction of rows $\by \sim \{0,1\}^d$ will have $\langle f(\by), g(x)\rangle > 0$ almost surely, and therefore, $A_{\by,\bx} = \langle f(\by), g(\bx)\rangle^2 > 0$. 

One challenge with lower bounding the rank of $A$ (which implies the lower bound on $s^2$) is that the assumption that $A_{y,x} > 0$ on entries $x, y$ has no margin (and could be arbitrarily close to zero). The first step of our argument involves using a Forster transform to create an average margin, where we will the Forster transform as in Lemma~4.19 of~\cite{HKLM20}.
\begin{lemma}\label{lem:fro-lb}
There exists a matrix $B \in [0,1]^{2^d \times 2^{d}}$ of rank $s^2$ satisfying $B_{x,z} = 0$ for every $z \in \calB(x,\tau)$, and $\| B \|_F^2 \geq \Omega((1-2\delta) 4^d / s^2)$.
\end{lemma}

\begin{proof}
Consider the subset of the half of columns $\calV \subset \{0,1\}^d$ (i.e., of size at least $2^d/2$), containing at least $2^d(1-2\delta)$ non-zeros in $A$. We will slightly overload notation by letting $V = \Span_{g}(\calV)$ to denote the span of vectors $\{ g(x)^{\otimes 2} : x \in \calV \}$, and we let $m \leq s^2$ be the dimensionality of $V$. The first step is identifying a subset $\calV' \subset \calV$ spanning a subspace $V'= \Span_{g}(\calV')$ of dimension $m' = \dim(V')$ satisfying the following two conditions (where the first is the assumption that will allows us to put the vectors corresponding to columns in $\calV'$ into radial isotropic position): 
\begin{enumerate}
\item\label{en:prop} Every $1 \leq k < m'$ and every subspace $\tilde{V} \subset V'$ of dimension $k$ has $|\{ x \in \calV' : g(x)^{\otimes 2} \in \tilde{V}| \leq (k/m') \cdot |\calV'|$.
\item\label{en:size} We have $|\calV'| \geq (m' / m) |\calV|$. 
\end{enumerate}
The fact such a subset exists follows an iterative argument. We let $\calV_0 = \calV$, $V_0 = V$ and $m_0 = m$ and will define a sequence of subsets parameterized by $\ell$. We maintain the invariant that $|\calV_{\ell}| \geq (m_{\ell} / m) |\calV|$ (which trivially holds for $\ell = 0$). Fix some $\ell$; if Condition~\ref{en:prop} fails to hold, we find $k$ between $1$ and $m_{\ell}-1$ and a subset $\tilde{V} \subset V_{\ell}$ where $|\{x \in \calV_{\ell} : g(x)^{\otimes 2} \in  \tilde{V}\}| > (k/m_{\ell}) |\calV_{\ell}|$, then we update $\calV_{\ell+1} = \{ x \in \calV_{\ell} : g(x)^{\otimes 2} \in \tilde{V}\}$, $V_{\ell+1} = \Span_g(\calV_{\ell+1})$, and $m_{\ell+1} = \dim(V_{\ell+1}) \leq k < m_{\ell}$, which now satisfies 
\[ |\calV_{\ell+1}| > (k/m_{\ell}) |\calV_{\ell}| \geq (k/m_{\ell}) \cdot (m_{\ell} / m) \cdot |\calV| \geq k/m \cdot |\calV| \geq m_{\ell+1} / m \cdot |\calV|. \]
Otherwise, if Condition~\ref{en:prop} is satisfied, we terminate. Since we always decrease the dimension $m_{\ell}$, the process must terminate at some $\ell$ and we set $\calV' = \calV_{\ell}$. Condition~\ref{en:size} follows from the invariant, and Condition~\ref{en:prop} holds when the process terminates.

We apply Lemma~4.19 of~\cite{HKLM20} to put vectors corresponding to columns from $\calV'$ into radial isotropic position. In other words, for any $\eps > 0$, we find an invertible linear map $R \colon V' \to V'$ which places $\calV'$ into $\eps$-isotropic position, that is, for any unit vector $u \in V'$, 
\begin{align} 
\frac{1-\eps}{m'} \leq \frac{1}{|\calV'|} \sum_{x \in \calV'} \dfrac{\langle u, Rg(x)^{\otimes 2} \rangle^2}{\| Rg(x)^{\otimes 2}\|_2^2} \leq \frac{1+\eps}{m'}.  \label{eq:eps-iso}
\end{align}
We may assume without loss of generality that $R$ is symmetric; if not, write $R = U \Lambda W^{\intercal}$ where columns of $U$ and $W$ are orthonormal bases of $V'$, and let $R' = W \Lambda W^{\intercal}$, which also satisfies (\ref{eq:eps-iso}) by writing $u = U W^{\intercal} w$ and varying over all unit vectors $w \in V'$.

For any $\alpha > 0$, we will define a symmetric invertible linear transformations $P_{\alpha} \colon \R^s \to \R^s$ to accomplish the following task:
\begin{itemize}
\item We will apply $P_{\alpha}$ to $g(x)^{\otimes 2}$ and re-normalize, and apply $P_{\alpha}^{-1}$ to $f(y)^{\otimes 2}$ and re-normalize. The fact that these are symmetric and inverses will preserve orthogonality, and the re-normalization will enforce inner products are in $[0,1]$.
\item On the subspace $V'$, the map $P_{\alpha}$ simply applies $R$. This let's us conclude that the vectors $g(x)^{\otimes 2}$ for $x \in \calV'$ will satisfy (\ref{eq:eps-iso}). Outside the subspace $V'$, the map $P_{\alpha}$ ``blows up'' vectors significantly, which has the effect that the inverse $P_{\alpha}^{-1}$ ``kills'' components of vectors outside $V'$, and effectively acts as an invertible projection onto $V'$.
\end{itemize}
Let $\Pi \colon \R^{s^2} \to V'$ and let $P_{\alpha}\colon \R^{s^2} \to \R^{s^2}$ be given by $R \Pi + (1/\alpha) \cdot (I - \Pi)$ which is also invertible and symmetric. Consider the collection of vectors $\{u_y^{\alpha}\}, \{ v_x^{\alpha} \}$ which will define the $2^d \times 2^d$ matrix $B^{\alpha}$ by letting $B_{y,x}^{\alpha} = \langle u_y^{\alpha}, v_x^{\alpha}\rangle$ and are given by
\begin{align*}
u_y^{\alpha} = \dfrac{P_{\alpha}^{-1} f(y)^{\otimes 2}}{\| P_{\alpha}^{-1} f(y)^{\otimes 2}\|_2} \in \R^{s^2} \qquad \text{and}\qquad v_x^{\alpha} = \dfrac{P_{\alpha} g(x)^{\otimes 2}}{\| P_{\alpha} g(x)^{\otimes 2}\|_2} \in \R^{s^2}
\end{align*}
Notice that $B^{\alpha}$ is an entry-wise positive re-scaling of $A$ which does not increase the rank. This is because 
\[ B_{y,x}^{\alpha} = \langle u_y^{\alpha}, v_x^{\alpha}\rangle = \dfrac{\langle P_{\alpha}^{-1} f(y)^{\otimes 2} , P_{\alpha} g(x)^{\otimes 2} \rangle}{\|P_{\alpha}^{-1} f(y)^{\otimes 2}\|_2 \cdot \| P_{\alpha} g(x)^{\otimes 2}\|_2} = \dfrac{\langle f(y)^{\otimes 2}, P_{\alpha}^{-1} P_{\alpha} g(x)^{\otimes 2}\rangle}{\|P_{\alpha}^{-\intercal} f(y)^{\otimes 2}\|_2 \cdot \| P_{\alpha} g(x)^{\otimes 2}\|_2}, \]
by symmetry of $P_{\alpha}^{-1}$, so the right-most numerator is exactly $A_{y,x}$ and the denominator is a positive rescaling. The fact that every entry $B_{y,x}^{\alpha} \in [0,1]$ is because it is an inner product of two unit vectors. 

It remains to show there is a choice of $\alpha$ which ensures $\| B^{\alpha}\|_F^2$ is large, and we will then let $B = B^{\alpha}$ for that choice of $\alpha$. For any fixed $y \in \{0,1\}^d$, we lower bound the norm of the $y$-th row using vectors in $\calV'$ (where $P$ acts exactly like $R$ since $g(x)^{\otimes 2} \in V'$ for $x \in \calV'$),
\begin{align*}
\| B_y^{\alpha} \|_2^2 \geq \sum_{x \in \calV'} \dfrac{\langle u_y^{\alpha}, P_{\alpha} g(x)^{\otimes 2}\rangle^2}{\| P_{\alpha} g(x)^{\otimes 2}\|_2^2} = \sum_{x \in \calV'} \dfrac{\langle \Pi u_y^{\alpha}, R g(x)^{\otimes 2}\rangle^2}{\| R g(x)^{\otimes 2}\|_2^2} \geq |\calV'| \cdot \dfrac{1-\eps}{m'} \cdot \| \Pi u_y^{\alpha} \|_2^2,
\end{align*}
where we first use the fact $Rg(x)^{\otimes 2} \in V'$ to project $u_y^{\alpha}$ to $V'$, and the last line used the $\eps$-isotropy of (\ref{eq:eps-iso}) on the unit vector $\Pi u_y^{\alpha} / \| \Pi u_y^{\alpha}\|_2$. Finally, notice that
\[ \| \Pi u_y^{\alpha}\|_2^2 = \dfrac{\|R^{-1} \Pi f(y)^{\otimes 2}\|_2^2}{\|R^{-1} \Pi f(y)^{\otimes 2}\|_2^2 + \alpha^2 \| (I-\Pi) f(y)^{\otimes 2}\|_2^2} \mathop{\longrightarrow}^{\alpha \to 0} \left\{\begin{array}{cc} 1 & \Pi f(y)^{\otimes 2} \neq 0 \\
0 & \text{o.w.} \end{array} \right. . \]
Note, the final condition, that $\Pi f(y)^{\otimes 2} \neq 0$ must occur whenever there is some $x \in \calV'$ with $\langle f(y)^{\otimes 2} , g(x)^{\otimes 2}\rangle > 0$, since $g(x)^{\otimes 2} \in V'$, i.e., when there exists some $x \in \calV'$ where $A_{y,x} \neq 0$. Thus,
\begin{align*}
\liminf_{\alpha \to 0} \| B^{\alpha} \|_F^2 \geq \frac{1-\eps}{m'} \cdot |\calV'|\sum_{y \in \{0,1\}^d} \ind\left\{ \exists x \in \calV' : A_{y,x} > 0 \right\} &\geq \frac{1-\eps}{m'} \sum_{y \in \{0,1\}^d} \sum_{x \in \calV'} \ind\left\{ A_{y,x} > 0\right\} \\
	&\geq \dfrac{1-\eps}{m'} \sum_{x \in \calV'} (1-2\delta) \cdot 2^d,
\end{align*}
where the last line used the fact every column $x \in \calV' \subset \calV$ has at least $(1-2\delta)$-fraction of entries $y \in \{0,1\}^d$ where $A_{y,x} > 0$. We can now lower bound this right-most expression by $(1-\eps) (1-2\delta) \cdot 2^d |\calV'| / m'$, and Condition~\ref{en:size} implies this is at least $(1-\eps) (1-2\delta) 4^d / (2m)$, and the final bound follows since $m \leq s^2$. As a result, there must exist some non-zero setting of $\alpha > 0$ giving the desired bound.
\end{proof}

The fact $B$ has all entries in $[0,1]$ implies that the sum of all entries in $B$ is at least the squared Frobenius norm. Suppose we let $M$ be any $2^d \times 2^d$ matrix which satisfies $M_{z,x} = 1$ whenever $z \notin \calB(x, \tau)$, then, we have the following chain of inequalities
\begin{align*}
\| B\|_F^{2} \leq \sum_{z,x \in \{0,1\}^d} B_{z,x} = \langle B, M \rangle \leq \| B \|_{\Tr} \cdot \| M \|_{\textrm{op}} \leq s \| B \|_F \cdot \|M \|_{\textrm{op}},
\end{align*}
where the upper bound on $\langle B, M\rangle$ in terms of the trace norm of $B$ and operator norm of $M$ is the matrix version of Holder's inequality with the $(1,\infty)$-pair, and the upper bound $\| B \|_{\Tr} \leq s \| B\|_F$ is because $B$ has rank at most $s^2$. Simplifying the squared Frobenius norm and substituting Lemma~\ref{lem:fro-lb}, we have the lower bound on $s^2$ for any small enough constant $\delta \in (0, 1/2)$: 
\[ s^2 \gsim \max \left\{\dfrac{2^d}{\|M\|_{\textrm{op}}} : M_{z,x} = 1 \text{ whenever $z \notin \calB(x,\tau)$} \right\}. \] 
Hence, it remains to show that there exists a matrix $M$ with $1$'s on entries $(z,x)$ whenever $z \notin \calB(x, \tau)$ which attains smallest possible spectral norm, which we do in the next two lemmas by reducing to a problem on extremal polynomials.

\begin{definition}\label{def:cr}
For $d, \tau \in \N$ with $\tau \leq d$, let $V(d, \tau)$ be the maximum value $q(0)$ over degree-$\tau$ polynomials where $|q(j)| \leq 1$ for $j \in \{1, \dots, d\}$.  
\end{definition}

\begin{lemma}\label{lem:ext-polynomial}
There exists a $2^d \times 2^d$ matrix $M$ where $M_{z,x} = 1$ whenever $z \notin \calB(x, \tau)$ and $\| M \|_{\textrm{\emph{op}}} \leq 2^{d} / (1 + V(d, \tau))$.
\end{lemma}

\begin{proof}
We will find a matrix $M$ given by finding coefficients $c_0, \dots, c_{d} \in \R$ and considering the factorization 
\[ M = 1_{2^d \times 2^{d}} - \sum_{\ell=0}^{d} c_{\ell} \sum_{\substack{S \subset [d] \\ |S| = \ell}} \chi_S \chi_{S}^{\intercal},  \qquad \text{s.t.} \qquad \begin{array}{ll} \text{(i)} & \alpha_{\ell}(j) = \sum_{|S| = \ell} \chi_{S}(x \oplus z) \text{ for $|x\oplus z| = j$}\\
								\text{(ii)} & \sum_{\ell=0}^d c_{\ell} \alpha_{\ell}(j) = 0 \text{ for all } j \in \{ \tau+1, \dots, d \} \end{array} . \]
Recall, the vectors $\chi_{S} \in \{-1,1\}^{2^d}$ for subsets $S \subset [d]$ and equivalently, their indicator vectors $1_S\in \{0,1\}^d$ are defined by letting the index $x \in \{0,1\}^d$ have $(\chi_{S})_{x} = (-1)^{\langle 1_S,  x_i\rangle}$, and the collection of all such vectors $S \subset [d]$ form an orthogonal basis of $\R^{2^{d}}$, and the all-1's matrix $1_{2^d \times 2^d} = \chi_{\emptyset} \chi_{\emptyset}^{\intercal}$. Thus, in the above factorization of $M$, the collection $\chi_S$ are eigenvectors and the eigenvalues are $2^d(1-c_0)$, and $-2^dc_1, \dots, -2^dc_d$. Furthermore, note that (i) is invariant to the specific choice of $x, z \in \{0,1\}^d$ as long as $|x \oplus z| = j$, and (ii) forces the constraint that $M_{z,x} = 1$ for $z \notin \calB(x,\tau)$; this can be seen as any $|x \oplus z| = j > \tau$ will have
\[ M_{z,x} = e_z^{\intercal} M e_x = 1 - \sum_{\ell=0}^d c_{\ell} \sum_{|S| = \ell} \chi_S(z) \chi_{S}(x) = 1 - \sum_{\ell=0}^d c_{\ell} \sum_{|S| = \ell} \chi_S(x \oplus z) = 1 - \sum_{\ell=0}^d c_{\ell} \alpha_{\ell}(j) = 1. \]
The desired spectral norm on $M$ is given by the following linear program, which we will upper bound by relating to polynomials
\begin{align}
\min_{c_0,\dots, c_d \in \R} 2^d \cdot \max\left\{ |1 - c_0|, |c_1|, \dots, |c_d| \right\} \qquad \text{s.t.} \qquad \sum_{\ell=0}^d c_{\ell} \cdot\alpha_{\ell}(j) = 0 \text{ for all $j \in \{\tau + 1, \dots, d \}$}.
\label{eq:coeff-prog}
\end{align}
The translation proceeds by considering the $(d+1) \times (d+1)$ matrix $K$ (for the Krawtchouk matrix), whose $(j,\ell)$-th entry is $\alpha_{\ell}(j)$. Then, the constraints in (\ref{eq:coeff-prog}) can be succinctly written as that of ensuring that the vector $p = K c \in \R^{d+1}$ has the bottom entries $j \in \{ \tau+1, \dots, d \}$ being exactly zero. We use the fact $K^2 = 2^d \cdot I$,\footnote{To see why, consider the following perspective. Let $F$ be a $2^d \times (d+1)$ matrix whose columns correspond to vectors $x \in \{0,1\}^d$ and $F_{x, \ell}$ is the indicator variable that $|x| = \ell$; as an operator, $F$ receives a specification $v$ of $(d+1)$ values and generates the symmetric function $Fv \colon \{0,1\}^d \to \R$ where $(Fv)(x) = v_{|x|}$, and if we let $D$ be the $2^d \times 2^d$ diagonal matrix which rescales by \smash{$D_{x,x} = \binom{d}{|x|}^{-1}$}, then $DFv$ is the symmetric function whose sums at various ``levels'' of the hypercube is given by $v$. Notice, letting $H$ be the $2^d \times 2^d$ Hadamard matrix $H_{y,x} = (-1)^{\langle x, y \rangle}$, we have $K = F^{\intercal} H DF$---as an operator, take $v$ and build the symmetric function whose ``levels'' sum to that value, evaluate Fourier coefficients with $H$ scaled by $2^d$, and sum them according to their levels. Then, $K^2$ applies this procedure twice; importantly, $DFF^{\intercal}$ is the $2^d \times 2^d$ block-diagonal matrix whose blocks correspond to vectors in the same level, and thus acts as the identity on symmetric functions. Since the Fourier coefficients of a symmetric function are also symmetric, $K^2 = F^{\intercal} HH DF = 2^d F^{\intercal} DF = 2^d I$.} in order to translate (\ref{eq:coeff-prog}) as that of optimizing over $p$ instead of $c$ (by multiplying by $K/2^d$ on both sides of $p = Kc$), leaving us with 
\begin{align}
\min_{p \in \R^{d+1}} \max\left\{ |2^d - (Kp)_0| , |(Kp)_{1}|, \dots, |(Kp)_{d}| \right\} \qquad \text{s.t.}\qquad p_{\tau+1} = \dots = p_{d} = 0. \label{eq:poly}
\end{align}
The final step comes from the fact that, as we vary the choice of $p \in \R^{d+1}$ with $p_{\tau+1} = \dots = p_d = 0$, the collection of vectors $Kp$ span the subspace of polynomial evaluations $(q(0), \dots, q(d)) \in \R^{d+1}$ of univariate polynomials $q$ of degree at most $\tau$. This is because, for any $j \in \{0,\dots, d\}$, $(Kp)_j = \sum_{\ell=0}^{\tau} p_{\ell} \alpha_j(\ell)$ may be expressed as
\begin{align*} 
\sum_{\ell=0}^{\tau} p_{\ell} \left(\sum_{k=0}^{\min\{ \ell, j \}} (-1)^k \binom{j}{k} \binom{d-j}{\ell-k} \right)   &= \sum_{\ell=0}^{\tau} p_{\ell} \left(\sum_{k=0}^{\ell} \dfrac{(-1)^k}{k! (\ell-k)!} \prod_{t=0}^{k-1} (j-t) \prod_{t=0}^{\ell-k-1} (d - j - t) \right)  \\
			&= \sum_{\ell=0}^{\tau} p_{\ell} \left( \gamma_{\ell}^{(\ell)} \cdot j^{\ell} + \gamma^{(\ell)}_{\ell-1} \cdot j^{\ell-1} + \dots + \gamma^{(\ell)}_{0} \cdot j^0 \right),
\end{align*}
where using the convention $0! = 1$ and because $\gamma_{\ell}^{(\ell)} = (-1)^{\ell} \sum_{k=0}^{\ell} 1/(k!(\ell-k)!) > 0$ (which means the $\tau+1$ polynomials $\alpha_j(\ell)$ for $\ell \in \{0, \dots, \tau\}$ are independent). 

In other words, (\ref{eq:poly}) asks to find a degree-$\tau$ univariate polynomial $\tilde{q} \colon \R \to \R$ which minimizes the maximum value of $|2^d - \tilde{q}(0)|, |\tilde{q}(1)|, \dots, |\tilde{q}(d)|$. We let $\tilde{q}(t) = 2^d / (1 + V(d,\tau)) \cdot q(t)$, where $q$ is the polynomial from Definition~\ref{def:cr}, where $2^d - \tilde{q}(0) = 2^d / (1+V(d,\tau))$ and $|\tilde{q}(\ell)| \leq 2^d / (1+V(d,\tau))$ for all $\ell$, giving the desired bound.
\end{proof}

\begin{lemma}\label{lem:cr}
$V(d, \tau) = \exp\left(\Omega(\tau^2/d)\right)$.
\end{lemma}

\begin{proof}
    We will utilize the univariate polynomial underlying the approximate degree bound for the $\mathsf{OR}$ function of~\cite{BCdWZ99, BT22}, which shows that the $\eps$-approximate degree of $\mathsf{OR}$ on $n$ input bits is $O(\sqrt{n \log (1/\eps)})$ for any $\eps \in [3^{-n}, 1)$ (see Theorem~14 in~\cite{BT22}). In particular, we consider the smallest possible $\eps \in [3^{-n}, 1]$ for which there is a univariate polynomial $f \colon \R \to \R$ of degree $\tau$ which satisfies $f(i) \in [1-\eps, 1 + \eps]$ for every $i \in \{1, \dots, d\}$ and $f(0) \in [-1-\eps, -1 +\eps]$; by Theorem~14 of~\cite{BT22}, we may take $\eps \leq \exp\left( -\Omega(\tau^2 / d)\right)$. We now consider the re-scaled polynomial $p(i) = -(f(i) - 1) / \eps$, which satisfies $|p(i)| \leq 1$ for all $i \in \{1, \dots, d\}$ and $p(0) \geq (2-\eps)/\eps = \exp(-\Omega(\tau^2/d))$.
\end{proof}

\ignore{
\begin{proof}
The construction considers two auxiliary polynomials 
\[ q(t) = \frac{(-1)^m}{m!} \prod_{j=1}^m \left(t - j\right) = \left\{\begin{array}{cl}  1 & t = 0\\ 0 & t \in \{1, \dots, m\} \\
					\binom{t}{m} & t \in \{m+1, \dots, d\} \end{array} \right. , \qquad C(t) = T_{s}\left( \dfrac{(d + m+1) - 2t}{d-m-1}\right),\]
where $T_s$ is the degree-$s$ Chebyshev polynomial of the first kind, and consider the degree-$(m+s)$ polynomial 
\[ P(t) = q(t) \cdot C(t) / \left(\frac{ed}{m}\right)^m. \] 
We first bound the values of $P(1), \dots, P(d)$. Note, $P(1)= \dots = P(m) = 0$ and for $m+1 \leq t \leq d$, $|P(t)|$ is at most $\max\{ |T_s(\xi)| : \xi \in [-1,1]\} \leq1$, since the denominator in $P(t)$ always overcomes $q(t)$, and the input to $T_s$ has numerator $(d+m+1)- 2t$ which always lies between $-(d - m- 1)$ and $d - m-1$ for $t \in \{m+1, \dots, d\}$. Then, as long as $m$ smaller than $d$ by a large constant factor and $m / d$ is much larger than $1/s^2$ by a constant factor,
\begin{align*}
P(0) = T_s\left(1 + \frac{2(m+1)}{d-m-1}\right) / \left(\frac{ed}{m}\right)^m &\geq \exp\left( \Omega(\sqrt{m/d} \cdot s) - O(m \log(d/m)) \right),
\end{align*}
where we use the fact $T_s(1 + \delta)$ grows like $\exp(\Omega(s \sqrt{\delta}))$ for small $\delta \geq c_0 / s^2$ for a large enough constant $c_0 > 1$. Consider setting $m = \lceil \zeta \tau^2 / d \rceil$ and $s = \lceil \tau/2 \rceil$ for $\zeta$ which we fix shortly (to $c_0/ \log^2(d/\tau)$ for a small $c_0$). Then, $m$ is a small constant factor of $d$ as long as $\zeta$ is smaller than a small constant, and $m/d = \zeta \tau^2 / d^2$ is larger than $4/\tau^2$ by a constant factor whenever $\tau \gg \sqrt{d} / \zeta^{1/4}$. In other words, we are left with a lower bound on $P(0)$ of 
\[ \exp\left( \Omega(\sqrt{\zeta}\cdot \tau^2 / d) - O(\zeta \tau^2 \log(d/\tau) / d) \right) = \exp\left( \Omega(\sqrt{\zeta} \cdot \tau^2 / d) \right), \]
for $\zeta$ being $c_0/\log^2(d/\tau)$ for a large enough polynomial factor, and the desired lower bound holds whenever $\tau \gg \sqrt{d} \cdot \polylog(d)$. But note that the desired lower bound becomes vacuous $\tau$ does not meet the above threshold.
\end{proof} }

\subsection{A Streaming Lower Bound from Theorem~\ref{thm:ln-sketch-lb}}

Combining Theorem~\ref{thm:ln-sketch-lb} with a slight modification of the streaming-to-sketching reduction of~\cite{KPSW25}, we obtain the dynamic streaming lower bound claimed in Theorem~\ref{thm:lb-intro}. The statement requires a technical assumption that the randomized streaming algorithm $\calA$ has bounded randomness complexity; namely, $\calA$ must use a finite amount of randomness, which may depend arbitrarily on $n$, $\Delta$, $d$, $c$ and $r$, but importantly \emph{not} on the underlying stream length. This technical condition arises in prior work of~\cite{LNW14, KPSW25} and is needed for the streaming-to-sketching reduction. 

\StreamLb*

We use the following streaming-to-sketching reduction, obtained from~\cite{KPSW25}. Although the above theorem is only stated only for $m=1$ in~\cite{KPSW25}, the extension to general $m\in\N$ is implicit in their proof; we defer additional details to Appendix \ref{sec: appendix-B}. We now prove Theorem~\ref{thm:lb-intro}.

\begin{lemma}[Theorem~4 of~\cite{KPSW25} applied with $\calD$ from Definition~\ref{def:hard}]\label{lem:stream-to-sketch}
For $n \in \N$ and $c > 1$, let $d$ be as in Definition~\ref{def:hard} and set $r = \sqrt{d} / (c\sqrt{3})$. Let $\delta$ denote the constant in Definition~\ref{def:r-c-n-diam}. 
\begin{itemize}
\item If $\calA$ is a randomized streaming algorithm solving $(r,c,n)$-diameter with bounded randomness complexity,
\item For any $m\in\N$, there is an $(s,O((m^2N)^{N/2}))$-linear sketch $(T,\Dec)$ with $s\le \calS^+(\calA, m) / \log (m+1)$ satisfying
\begin{align}
\Prx_{\bv\sim \calD}\left[ \Dec(T\bv) \neq \diam_{r,c,n}(\bv) \right] \leq O(\delta).
\end{align}
\end{itemize}
\end{lemma}

\begin{proof}[Proof of Theorem~\ref{thm:lb-intro}]
    Let $\calA$ be a randomized streaming algorithm solving $(r,c,n)$-diameter with bounded randomness complexity. By Lemma~\ref{lem:stream-to-sketch}, there is a $(s,O((m^2N)^{N/2}))$-linear sketch $(T,\Dec)$ with $s \le \calS^+(\calA,m) / \log(m+1)$ with success probability $1-O(\delta)$ on $\calD$. Theorem~\ref{thm:ln-sketch-lb} implies $s\ge n^{\Omega(1/ (c^2\log^2 c))}$, so we conclude that $\calS^+(\calA,m)\ge n^{\tilde{\Omega}(1/c^2)} \cdot  \log(m+1)$.
\end{proof}
\section{Upper Bound} \label{sec: upper-bound}

In this section, we show how the insertion only streaming algorithm for Euclidean diameter estimation in~\cite{I03} can be adapted to the dynamic streaming setting. In particular, we show the following

\begin{theorem} \label{thm: upper-bound}
    There exists a dynamic streaming algorithm $\calA$ that solves the $(r, c, n)$-diameter problem with $\calS^+(\calA,m) = \tilde{O}(n^{2/(c^2 - 1)} \cdot \log m)$ bits of space for all $m \in \N$.\footnote{We assume for convenience, that the $n$ point sub-metric has aspect ratio $\poly(n)$ (to support aspect ratio $\Delta$, we pay an additive $\polylog(n\Delta)$ factor in the bit-complexity).}
\end{theorem}

Our proof will follow by first exhibiting an embedding of $\ell_2$ to $\ell_\infty$, then utilizing a dynamic streaming algorithm for estimating diameter in $\ell_\infty$.

\subsection{Preliminaries}

\begin{lemma} [Proposition 2.1.2 of \cite{V19}] \label{lem: gaussian-tails}
    Let $\bg \sim \calN(0,1)$. Then, for all $t > 0$, we have 
    \begin{align*}
        \left( \frac{1}{t} - \frac{1}{t^3} \right)\cdot \frac{2}{\sqrt{2 \pi}} \cdot e^{-t^2 / 2} \leq \Pr \left[|\bg| \geq t \right] \leq \frac{1}{t} \cdot \frac{2}{\sqrt{2 \pi}} \cdot e^{-t^2 / 2}
    \end{align*}
\end{lemma}

\begin{lemma}[Modified Corollary~5.3 of~\cite{KPSW25}] \label{lem: ell-infty-diameter}
    Fix any $r,\Delta > 0$ and any $c \ge 1 + \eps$ for $\eps, \delta \in (0,1)$, there is a dynamic streaming algorithm for $\diam_{U}^{r, c}$ over $\ell_{\infty}^k$ that uses a linear sketch with $\tilde{O}(k \log n \log(k / \delta\eps) / \eps)$ rows and $n$ columns. The total bit-complexity is $\tilde{O}(( k \log n \log(nm) + \log \Delta) \log(k/\delta\eps) / \eps)$, where $\Delta$ upper bounds the aspect ratio, and $m$ upper bounds the magnitude of the final frequency vector.
\end{lemma}

\begin{remark} \upshape
    Corollary~5.3 in~\cite{KPSW25} is stated for $c\ge 2(1+\eps)$ rather than $c\ge 1+\eps$. It loses a factor $2$ because their algorithm first samples an arbitrary surviving point $q$ and then estimates the maximum distance from $q$ to a surviving point, which may only be half of the diameter.
    
    We can avoid this loss by estimating the diameter along each coordinate separately and taking the maximum. For each coordinate, we partition the line into intervals of width $\eta = \eps r / 3$. Let $B = 2 + 3(1+\eps) / \eps$, and maintain $L = \Theta(B(\log B + \log(k/\delta))) = \Theta(\log(k/\delta\eps) / \eps)$ independent $\ell_0$-samplers to sample intervals that are occupied at the end of the stream. By a standard coupon-collector argument, the following holds for a given coordinate with probability $1-\delta/k$:
    \begin{itemize}
        \item If more than $B$ intervals are occupied, the samplers return more than $B$ intervals, certifying that the diameter is at least $(1+\eps)r$.
        \item If at most $B$ intervals are occupied, the samplers recover all of them. Then, the two extreme intervals recover the diameter up to additive error $\eta$, which allows us to determine whether it is $\le r$ or $\ge (1+\eps)r$.
    \end{itemize}
\end{remark}

\subsection{Proof of Theorem \ref{thm: upper-bound}}

For $k,d\in\N$, let $\calG(k,d)$ denote the distribution over matrices $\bG\in\R^{k\times d}$ where each coordinate $\bG_{ij}$ is drawn independently from $\calN(0,1)$.

\begin{lemma} \label{lem: ell-infty-embedding}
For any $n\in\N$ and $1 \leq c \leq \sqrt{\log n}$, there exists $k=\tilde{O}(n^{2 / (c^2-1)})$ and a scale $\rho>0$ so that the following holds. For any $X = {x_1,\ldots,x_n}\subset\R^d$, there is a mapping $\boldf : \R^d \to \R^k$ given by $\boldf(x) = \frac{1}{\rho} \bG x$ for $\bG\sim\calG(k,d)$ which satisfies
\[        \|x_i-x_j\|_2 \le \|\boldf(x_i) - \boldf(x_j)\|_\infty \le c \cdot \|x_i-x_j\|_2
        \quad \text{for all } (i,j)\in [n]^2
   \]
with probability at least $0.9$.
\end{lemma}

\begin{proof}
Fix a pair of points $x_i,x_j\in X$ and let $v = x_i-x_j$. We would like to show that
\(        \|v\|_2 \le \frac{1}{\rho} \|\bG v\|_\infty \le c \cdot \|v\|_2
   \)
holds with probability $1 - 1/(10n^2)$. We first observe that each coordinate of $\bG v \in\R^k$ is distributed as $\calN(0,\|v\|_2^2)$. Dividing out by $\|v\|_2$, it suffices to show the following statement: for i.i.d. variables $\by_1,\ldots,\by_L\sim\calN(0,1)$, the maximum
\(        \by = \max_{\ell\in[k]} |\by_\ell|
   \)
satisfies $\by\in[\rho, c\rho]$ with probability $1-1/(10n^2)$.

We set $\rho = \frac{1}{c}\sqrt{2\ln(40Ln^2)}$. We will first argue that $\Pr[\by < \rho] \le 1/(20n^2)$. Note that
\begin{align*}
    \Pr[\by < \rho] = \left( 1 - \Pr[|\by_1| \geq \rho]\right)^{k} \leq \exp \left(-k \cdot \Pr[|\by_1| \geq \rho] \right)
\end{align*}
By Lemma \ref{lem: gaussian-tails}, we have that
$
    \Pr[|\by_1| \geq \rho] \geq \frac{2}{\sqrt{2 \pi}} \cdot \left(\frac{1}{\rho} - \frac{1}{\rho^3} \right) \cdot e^{-\rho^2 / 2} 
$. Observe that since $k \geq 1$ and $c < \sqrt{\log n}$, we have that $\rho > 1$. Therefore, 
 $
     \Pr[|\by_1| \geq \rho] \geq \frac{\alpha}{\rho} \cdot e^{-\rho^2 / 2} 
$
for some constant $\alpha > 0$.

We set the parameter
\begin{align*}
    k := \left\lceil 10 n^{2/(c^2-1)} (\log n) ^{\frac{4c^2}{(c^2-1)}} \right\rceil 
\end{align*}
Then, $\rho = \tilde{O}(\sqrt{\log n})$ and thus
\begin{align*}
    k \cdot \Pr[|\by_1| \geq \rho] \geq k \cdot (Ln^2)^{-1/c^2} \cdot \frac{\alpha}{\rho} \geq k^{(c^2 -1)/c^2} \cdot n^{-2/c^2} \cdot \frac{\alpha}{\rho} 
    \geq \tilde{\Omega}(\log^2 n).
\end{align*}
For large enough $n$ then, we get that 
\begin{align*}
    \Pr[\by < \rho] \leq \exp \left(-k \cdot \Pr[|\by_1| \geq \rho] \right) \leq 1/(20n^2).
\end{align*}
In the other direction, we have by Lemma \ref{lem: gaussian-tails} that 
$
    \Pr[\by \geq c \rho] \leq k \cdot \Pr[|\by_1| \geq c \rho] \leq \frac{\beta k}{\rho} \cdot e^{- c^2 \rho^2 /2}
$
for some constant $\beta > 1/2$. Therefore, 
$
    \Pr[\by \geq c \rho] \leq 1/(20n^2)
$
and we can conclude the lemma. 
\end{proof}

We are now ready to prove Theorem \ref{thm: upper-bound}. 

\begin{proof}[Proof of Theorem \ref{thm: upper-bound}]
    Let $C>1$ denote the desired approximation factor. Note that if $C \le \sqrt{3}$ then $n^{2/(c^2-1)} = \Omega(n)$ and we can compute the diameter exactly by storing the frequency vector. We can therefore assume $C > \sqrt{3}$.

    For $\gamma = \Theta(1/\log n)$, we apply Lemma~\ref{lem: ell-infty-embedding} to obtain an embedding $f : U \to \ell_\infty^K$ with distortion $\alpha = C / (1+\gamma)$, where \[K = \tilde{O}(n^{2/(\alpha^2 - 1)}) = \tilde{O}(n^{2/(C^2 - 1)}),\] which follows because $1/(\alpha^2-1) = 1/(C^2-1) + O(1/\log n)$. We then run the streaming algorithm of Lemma \ref{lem: ell-infty-diameter} with approximation $(1+\gamma)$, whose bit complexity is $\tilde{O}(K\log\Delta / \gamma) = \tilde{O}(n^{2/(C^2-1)})$ as desired.
\end{proof}

\section*{Acknowledgements}

The authors would like to thank Sanjeev Khanna and Josh Alman for many helpful discussions throughout the course of this project. Part of this work was aided by Claude Fable 5.1 to navigate the literature. In particular, Fable 5.1 was very useful for explaining the Forster transform, pointing us to~\cite{HKLM20}, and sketching the proof of Lemma~\ref{lem:fro-lb} and Lemma~\ref{lem:ext-polynomial}. This writeup was produced by the authors and we take full responsibility for the content.

\bibliographystyle{alpha} 
\bibliography{waingarten} 

\appendix 


\section{Proof of Lemma~\ref{lem:alg-to-embedding}} \label{sec: appendix}

In this appendix, we exhibit settings of the parameters $\sfP, \sfU$, along with a third parameter $\sfQ$ that will be useful in the analysis, for distribution $\calD$ to prove Lemma \ref{lem:alg-to-embedding}. We assume that $(T, \Dec)$ is an $(s, \lambda)$-linear sketch. Recall that $N := |\{ 0, \ldots, \Delta\}^d| = (\Delta+1)^d$. 

\AlgToEmbedding*

The settings of $\sfP, \sfU, \sfQ$ we will use are the following. 

\begin{equation} \label{eq: parameters}
    \sfP := \del{O(s \lambda)^{2s}}!
    \qquad
    \sfU := \omega\!\del{ (n-1) \cdot O(s\lambda)^{2s+2}}
    \qquad
    \sfQ := \sfP \cdot O(s\lambda)^{2s + 2} 
\end{equation}

\paragraph{Forbidden Kernel Vectors.} \label{subsec: forbidden} We now have a complete description of $\calD$ with parameters $\sfP, \sfU$. We assume that $(T, \Dec)$ is such that 
\begin{align*}
    \Prx_{\bv \sim \calD}[\Dec(T \bv) \neq \diam_{r, c, n}(\bv)] \leq \delta
\end{align*}
and show how to construct an $(O(\delta), \tau)$-ball-subspace-evading function $f: \hypercube \to \R^s$.  

We first prove the crux of the argument: for almost all pairs $(\bx,\by)$ drawn from $\calD$, the \emph{kernel} of $T$ contains no ``short vectors'' supported on $\calB(\bx,  \tau) \cup \{\by\}$ that place multiplicity $\sfP$ on $y$. Such a vector would let us toggle whether $\by$ is included without changing the output of the sketch, which would contribute to the error on $\calD$.

\begin{claim} \label{claim: forbidden-kernel-vectors}
    For $x, y \in \hypercube$, let $Z_{x,y}(\sfP)$ be the set of vectors $z \in \Z^N$ satisfying
    \begin{enumerate}[label=(\roman*)]
        \item $z_y = \sfP$,
        \item $\|z\|_\infty \le \sfQ$,
        \item $z_u = 0$ for every $u \notin \calB(x, \tau) \cup \{y\}$.
    \end{enumerate}
    Then for at least a $(1 - 4\delta)$ fraction of pairs $(x, y) \in \hypercube \times \hypercube$, we have $\ker(T) \cap Z_{x,y}(\sfP) = \emptyset$.
\end{claim}
\begin{proof}
    Consider the following two sets 
    \begin{align*}
        S &:= \{(x,y) \in \hypercube \times \hypercube: \ker(T) \cap Z_{x,y}(\sfP) \ne \emptyset\} \\ \qquad L &:= \{(x,y) \in \hypercube \times \hypercube: \Ham(x,y) \ge d/3\}
    \end{align*}
    We suppose for contradiction that $|S| > 4 \delta \cdot 2^{2d}$. For $d$ a large enough constant, note that $|L| \ge (1 - \delta) 2^{2d}$, which implies $|S \cap L| > 3 \delta \cdot 2^{2d}$. For each $(x,y) \in S \cap L$, fix a \emph{witness} $$z^{x,y} \in \ker(T) \cap Z_{x,y}(\sfP)$$ Additionally, for each $x \in \hypercube$, we define its \emph{level set} $$Y_x := \{ y : (x,y) \in S \cap L\}$$ and let $m_x := |Y_x|$. For a fixed $x$, call a frequency vector $v \in \Z_{\ge 0}^N$ \emph{$x$-good} if $\supp(v) = \calB(x,\tau)$ and $v_u \in \{\sfQ + 1, \dots, \sfP\sfU - \sfQ\}$ for all $u \in \calB(x,\tau)$. Each of the $(\sfP\sfU - 2\sfQ)^{n-1}$ such vectors is sampled by $\calD$ with probability $\frac{1}{2^{d+1}}(\sfP\sfU)^{-(n-1)}$.

    For any $x$-good vector $v$ and $y \in Y_x$, the vector $v + z^{x,y}$ is bounded within $\{1, \dots, \sfP\sfU\}$ on $\calB(x,\tau)$ and takes value $\sfP$ at $y$. Because $(x,y) \in L$, $y \notin \calB(x,\tau)$, meaning $\diam_{r, c, n}(v + z^{x,y}) = 1$ while $\diam_{r, c, n}(v) = 0$. Since $z^{x,y} \in \ker(T)$, however, we have $T(v + z^{x,y}) = Tv$. Thus the linear sketch returns the same value for all inputs in $\{v\} \cup \{v + z^{x,y} : y \in Y_x\}$. Hence, these vectors contribute error at least
    \[
        \min\cbr{\frac{1}{2^{d+1}} (\sfP\sfU)^{-(n-1)},\, m_x \cdot \frac{1}{2^{2d + 1}}(\sfP\sfU)^{-(n-1)}} = \frac{m_x}{2^{2d + 1}}(\sfP\sfU)^{-(n-1)},
    \]
    using the fact that $m_x \le N$. Crucially, these blocks are disjoint across all $x$ and all $x$-good vectors, as $\calB(x,\tau)$ and $\calB(x',\tau)$ are distinct when $x\neq x'$. Thus, we can sum the error over all $x$, obtaining a total error of
    \[
        \sum_{x\in\hypercube} (\sfP\sfU - 2\sfQ)^{n-1} \frac{m_x}{2^{2d+1}}(\sfP\sfU)^{-(n-1)} = \left( 1 - \frac{2\sfQ}{\sfP\sfU} \right)^{\!n-1} \frac{|S \cap L|}{2^{2d+1}} .
    \]
    By~\eqref{eq: parameters}, we have $(1 - \frac{2\sfQ}{\sfP\sfU})^{n-1} \ge 0.99$. The total error is therefore at least $\frac{0.99}{2} \cdot 3 \delta > \delta$, which contradicts the guarantee of~\eqref{eq:comp} from Lemma~\ref{lem:alg-to-embedding} that linear sketch $(T, \Dec)$ has failure probability $\leq \delta$ on distribution $\calD$.
\end{proof}

\paragraph{Constructing a Ball-Subspace-Evading Function.} \label{subsec: evasion}

Using Claim~\ref{claim: forbidden-kernel-vectors}, we show that with high probability over a draw $(\bx,\by)$, the $\by$-column of $T$ does not lie in the span of the columns corresponding to $\calB(\bx,\tau)$. 

\begin{claim} \label{claim: far-point-outside-span}
    Let $(T, \Dec)$ be an $(s, \lambda)$-linear sketch with failure probability at most $\delta$ on distribution $\calD$, as in Lemma \ref{lem:alg-to-embedding}. Then,
    \[
        \Prx_{\bx, \by \sim \hypercube} \left[T^{(\by)} \notin \linspan\del{\{ T^{(u)} : u \in \calB(\bx, \tau) \}}\right] \geq 1 - 4\delta.
    \]
\end{claim}
\begin{proof}
    By Claim~\ref{claim: forbidden-kernel-vectors}, at least $(1 - 4 \delta)$-fraction of pairs $(x,y) \in \hypercube \times \hypercube$ satisfy $\ker(T) \cap Z_{x,y}(\sfP) = \emptyset$. We show that every such pair satisfies the above subspace-evasion property. Fix any such pair $(x,y)$ and suppose for contradiction that $T^{(y)} \in \linspan\{ T^{(u)} : u \in \calB(x,\tau)\}$.

    Let $A\in\Z^{s\times (n-1)}$ denote the submatrix of $T$ consisting of the columns corresponding to $\calB(x,\tau)$. Let $W\in\Z^{s\times k}$ denote the submatrix of $A$ which contains a maximal set of linearly independent columns. We will show that there is no vector $w\in\R^{n-1}$ such that $Aw = T^{(y)}$.

    Suppose for the sake of contradiction that such a $w \in \R^{n-1}$ exists. Then, there would also exist $y \in \R^{k}$ which satisfies $Wy = T^{(y)}$. Because $T^{(y)}$ and $W$ have integer coordinates of magnitude at most $\lambda$, there exists an integer vector $\tilde{a} \in \Z^k$ which satisfies $W \tilde{a} = \sfD \cdot T^{(y)}$, where $\sfD := \det(W^{\t} W)$ is an integer of magnitude at most $O(s \lambda)^{2s}$, and $\| \tilde{a} \|_{\infty} \leq \sfD \cdot O(s\lambda)^{2s+2}$.
    
    Thus, the setting of $\sfP = (O(s \lambda)^{2s})!$ implies that $\sfP / \sfD$ is an integer, and $\tilde{a}' = \sfP / \sfD \cdot \tilde{a}$  satisfies $W \tilde{a}' = \sfP \cdot T^{(y)}$ and $\| \tilde{a}'\|_{\infty} \leq \sfP\cdot O(s \lambda)^{2s+2}$. Let $z \in \mathbb{Z}^{N}$ be defined as $z_y = \sfP$ and $z_u = -\tilde{a}_{u'}'$ for each $u$ such that $u \in \calB(x, \tau)$, where $\tilde{a}_{u'}'$ is the coordinate of $\tilde{a}'$ associated with the column of $T$ indexed by $u$.
\end{proof}

We are now ready to prove Lemma \ref{lem:alg-to-embedding}. 

\begin{proof}[Proof of Lemma~\ref{lem:alg-to-embedding}]
    Let $\sfP, \sfU$ be set as in (\ref{eq: parameters}) and let $(T, \Dec)$ be an $(s, \lambda)$-linear sketch with failure probability at most $\delta$ on distribution $\calD$, as in Lemma \ref{lem:alg-to-embedding}. Consider the embedding $f \colon \hypercube \to \R^s$ given by $f(u) := T^{(u)}$. For any $x \in \hypercube$,
    \[
        \linspan_f(\calB(x, \tau)) = \linspan\bigl( \{ T^{(u)} : u \in \calB(x, \tau) \} \bigr),
    \]
    so Claim~\ref{claim: far-point-outside-span} gives exactly
    \[
        \Prx_{\bx, \by \sim \hypercube}\left[ f(\by) \notin \linspan_f(\calB(\bx, \tau)) \right] \geq 1- 4 \delta,
    \]
    which is the conclusion of Lemma~\ref{lem:alg-to-embedding}.
\end{proof}

\section{Streaming to Sketching via~\cite{KPSW25}} \label{sec: appendix-B}

We describe how Theorem~4 of~\cite{KPSW25} implies Lemma~\ref{lem:stream-to-sketch}. Theorem~4 of~\cite{KPSW25} states the following: \begin{itemize}
    \item If $\calA$ is a randomized streaming algorithm with bounded randomness complexity that computes a scale-invariant function $g$ over all dynamic streams with probability $1-\delta$,
    \item Then for any distribution $\calD$ supported on $\Z^N\cap g^{-1}(\{0,1\})$, there exists an $(s, O(N^{N/2}))$-linear sketch $(T,\Dec)$ with $s\le \calS^+(\calA,1)$ where \[\Prx_{\bv\sim\calD}[\Dec(T\bv)=g(\bv)] \ge 1-O(\delta).\]
\end{itemize}

Scale-invariant functions $g$ are defined those where $g(\lambda v) = g(v)$ for any $\lambda\in\N$. In particular, the diameter decision function $\diam_{r,c,n}$ is scale-invariant since it only depends on the support of the frequency vector $v$.

From this, we can immediately conclude the $m=1$ case of Lemma~\ref{lem:stream-to-sketch}.\footnote{While Theorem~4 of~\cite{KPSW25} requires that $\diam_{r,c,n}(v)\in\{0,1\}$ for all $v\in\supp(\calD)$, we note that the distribution $\calD$ from Definition~\ref{def:hard} only satisfies $\diam_{r,c,n}(\bv)\in\{0,1\}$ with high probability over $\bv\sim\calD$. However, one can simply apply Theorem~4 on the appropriate conditional distribution, recovering the same guarantee with an $o(1)$ loss in the failure probability.} The extension to general $m\in\N$ is not immediate, but is implicit in the proof of Theorem~4. The proof follows the following roadmap:\begin{itemize}
    \item \textbf{Path-independence:} First, it shows that there exists a \emph{path-independent} algorithm $\calB$ with space complexity $\calS^+(\calB,m)\le \calS^+(\calA,m)$ for all $m\in\N$.
    \item \textbf{Constructing a basis:} Second, via Lemma~2.12 in~\cite{KPSW25}, it constructs a basis $b_1,\ldots,b_N$ of $\R^N$ where the final $N-t \le \calS^+(\calB,1)$ basis vectors have entries bounded by $N^{N/2}$ and span the orthogonal complement of the \emph{zero-set subspace} of $\calB$.
    \item \textbf{Building the sketch:} Finally, it builds a sketching matrix $T$ by taking the vectors $b_{t+1},\ldots,b_N$ as the rows, and then designs an appropriate decoding function. By construction, the sketch $T$ has dimension at most $\calS^+(\calB,1) \le \calS^+(\calA,1)$ and entries bounded by $O(N^{N/2})$. 
\end{itemize}

The only part of the argument specific to $m=1$ is the second step via Lemma~2.12. However, it is clear that the statement extends to general $m$, only affecting the final sketch dimension and matrix entry bound.

\begin{itemize}
    \item \textbf{Sketch dimension:} By Lemma~C.4 of~\cite{KPSW25}, one has $\calS^+(\calB,m) \ge (N-t)\log(m+1)$ for any $m\in\N$. This immediately gives $s = N-t \le \calS^+(\calB,m) / \log(m+1)$.
    \item \textbf{Matrix entry bound:} Originally, the upper bound on $\|b_{t+1}\|_\infty, \ldots, \|b_N\|_\infty$ is given by the determinant of a $t\times t$ matrix with entries in $\{-1,0,1\}$, which is at most $t^{t/2} \le N^{N/2}$. For general $m$, we instead take the determinant of a matrix with entries in $\{-m,\ldots,m\}$, giving the bound $(m^2 N)^{N/2}$.
\end{itemize} 

\end{document}